\documentclass[12pt]{amsart}

\usepackage{amssymb,amscd}

\usepackage{enumerate}

\usepackage{graphicx}

\usepackage[all]{xy} 
\usepackage{mathtools} 

\usepackage{color} 

\usepackage[figuresright]{rotating}

\usepackage[]{breqn} 

\makeatletter
\@namedef{subjclassname@2010}{%
  \textup{2010} Mathematics Subject Classification}
\makeatother

\newtheorem{thm}{Theorem}[section]
\newtheorem*{thm*}{Theorem}
\newtheorem{cor}[thm]{Corollary}

\newtheorem{propo}[thm]{Proposition}

\theoremstyle{definition}
\newtheorem{defin}[thm]{Definition}
\newtheorem{rem}[thm]{Remark}

\numberwithin{equation}{section}

\def\M{\mathcal{M}}

\def\R{\mathcal{R}}
\def\RR{\mathbb{R}}
\def\Z2{\mathbb{Z}/2 \mathbb{Z}}

\def\Tp{\mathcal{T}_p} 
\def\Tx{\bigotimes^p T_x^*X}

\def\d{\mathrm{d}}

\def\qed{\hfill $\square$}

\def\proof{\noindent \textit{Proof: }}

\def\Diff{\protect \mathrm{Diff}_{p}}

\def\Ogx{\mathrm{O}_{g_x}}

\newcommand{\Hom}[3]{\ensuremath{{\protect \mbox{Hom}_{\scriptsize
#1}(#2,#3)}}}

\newcommand{\Ric}{\ensuremath{\protect{\mbox{Ric}}}} 

\renewcommand{\div}{\mathrm{div }\,}

\begin{document}

\newcommand{\dxdx}[2]{\begin{minipage}{2cm} $ \d x_{#1}\,\otimes \, \d x_{#2} $\end{minipage}}

\baselineskip=17pt

\title[Electromagnetic energy tensor in second order Lovelock gravities]{On the existence of the electromagnetic energy tensor in second order Lovelock gravities}
\author{R. Mart\'inez-Boh\'orquez}
\address{Departamento de Matem\'{a}ticas \\ Universidad de Extremadura \\ E-06071 Badajoz, Spain}
\email{raulmb@unex.es}
\thanks{The author has been partially supported by Junta de Extremadura and FEDER funds with projects IB18087, GR18001 and GR21055, and has been supported by the grant ``Plan 
Propio de Iniciación a la Investigación, Desarrollo Tecnológico e Innovación'' of Universidad de Extremadura.}

\date{\today}

\begin{abstract}
In this work we prove that there are no electromagnetic energy tensors in second order Lovelock gravities that verify properties equivalent to those of the Maxwell electromagnetic energy tensor in General Relativity.
\end{abstract}

\maketitle

\tableofcontents

\section{Introduction}

Electromagnetic fields on Lorentzian manifolds $(X,g)$ 
 can be represented by differential 2-forms $F$ on $X$ that verify Maxwell's equations. The electromagnetic energy tensor is then constructed with the data of both the metric $g$ and the 2-form $F$ in a natural way, meaning that the resulting tensor should not depend on the choice of coordinates on the manifold. 

The electromagnetic energy tensor in General Relativity is the well-known Maxwell tensor. In a local system of coordinates, it is given by the expression\footnote{The standard Einstein summation convention will be used throughout this text.}
\begin{equation}\label{EMT}
E_{ij}:=F_{ik} F_j^{\ k} - \frac{1}{4}F^{kl}F_{kl}g_{ij} \ .
\end{equation}

Many characterizations of the Maxwell tensor have been produced in the past; particularly relevant are those by D. Lovelock (\cite{LOVELOCK,LOVELOCK2,LOVELOCK3,LOVELOCK4}), I.M. Anderson (\cite{ANDERSON}) and D.B. Kerrighan (\cite{KERRIGHAN}). In recent times, a description involving the dependence on the unit of time-length has been given by Navarro-Sancho (\cite{NS_ELECTRO}):

\begin{thm*}[\cite{NS_ELECTRO}] The Maxwell energy tensor $E$ is the only 2-covariant tensor $T = T(g, F)$ naturally associated to a Lorentzian metric $g$ and a 2-form $F$, satisfying the following properties:
\begin{enumerate}
\item At any point, $F_x=0$ implies $T_x=0$.
\item If $\mathrm{d} F=0$, then $\div T= -i_{\partial F} F$, where $\partial F$ denotes the codifferential of $F$.
\item $T$ is independent of the unit of scale; that is, $T(\lambda^2 g, \lambda F) = T(g, F)$ for any $\lambda > 0$.
\end{enumerate}
\end{thm*}

These properties possess physical meaning: condition $(1)$ forces the dependence of the electromagnetic energy tensor on the existence of electromagnetic field, and condition $(2)$ is due to the Lorentz force law, assuring that the total energy-momentum tensor (the sum of the contributions of matter and electromagnetism) is zero, thus preserving the principle of conservation of energy.

Condition $(3)$ differentiates this characterization of the Maxwell tensor from the aforementioned classical results. It is an homogeneity condition that originates from the fact that the electromagnetic energy tensor is coupled with the Einstein tensor, and thus the dimensional units should be coherent (see \cite{NS_ELECTRO} for further details).

In this work, we study whether an electromagnetic energy tensor can be defined in second order Lovelock gravities, verifying properties equivalent to the above. 

Recall that Lovelock gravities are generalizations of General Relativity to higher dimensions, where other tensors (known as the Lovelock tensors $L_k$) of second order in the derivatives of the metric appear as possible candidates to describe the distribution of matter in dimensions greater than $2k$. These tensors qualitatively differ from the Einstein tensor in their degree of homogeneity: whereas the Einstein tensor is linear in the second derivatives of the metric, higher order Lovelock tensors  are homogeneous polynomials of degree $k$ in the second derivatives of the metric. 

Precisely, a change of the unit of time-length by a factor of a non-zero $\lambda \in \RR$ (i.e. a constant rescaling of the metric $g$ by a factor $\lambda^2$) modifies the $k$-th Lovelock tensor $L_k$ by a factor of $\lambda^{2-2k}$:
$$L_k(\lambda^2 g)=\lambda^{2-2k} L_k(g) \ .$$
In other words, the $k$-th Lovelock tensor is homogeneous of weight $2-2k$.

 The simplest Lovelock tensor after the metric tensor itself ($L_0$) and the Einstein tensor ($L_1$) is the Lovelock tensor $L_2$ of order $2$, also known as the Lanczos-Lovelock tensor:
$$(L_2)_{ij} = - R_{iabc} R_{j}^{\ abc} + 2\Ric^{ab} R_{iajb} + 2\Ric_{ia} \Ric^a_{\ j} - r \Ric_{ij} + \frac{1}{4} \left( R_{abcd} R^{abcd} -4 \Ric_{ab} \Ric^{ab} +r^2  \right) g_{ij} \ ,$$ 

where $R$ denotes the Riemann-Christoffel curvature tensor, $\Ric$ denotes the Ricci curvature tensor and $r$ denotes the scalar curvature.

We prove (Theorem \ref{MainTheorem} and its corollary) that there exists no natural 2-tensor associated to metrics and 2-forms that verifies the following properties: 
\begin{enumerate}
\item If $F_x=0$, then $T_x=0$.
\item $\div T = -i_{\partial F} F$ whenever $F$ is closed, where $\partial F$ denotes the codifferential of $F$.
\item The dimensional units of $T$ are coherent; that is to say, that they coincide with those of the Lanczos-Lovelock tensor $L_2$. In mathematical terms, it means that $T$ should be homogeneous of relative weight $(1,-2)$, that is, for any non-zero $\lambda \in \RR$ it holds that\footnote{Recall that a change of the unit of time-length by a factor of $\lambda$ also modifies the electromagnetic field $F$ by a factor of $\lambda$ (see \cite{NS_ELECTRO}).}
$$T(\lambda^2 g, \lambda F)=\lambda^{-2} T(g,F) \ .$$
\end{enumerate}

%

Therefore, an electromagnetic energy tensor $T$ which can be coupled to the Lanzcos-Lovelock tensor that verifies the corresponding properties to the Maxwell tensor in General Relativity does not exist.

In order to prove this statement, we will make use of several deep results of invariant theory, which can be found in the essential monograph by Kol\'a\v{r}-Michor-Slov\'ak (\cite{KMSBOOK}) and were rewritten to a more modern language by Navarro-Sancho (\cite{N_THESIS,NS_EINSTEIN}). Invariant theory has long been utilized in the research of natural constructions and field equations in Physics, on of the first examples being deeply rooted in the origins of General Relativity, where the fundamental characterization of the Einstein tensor was obtained by Vermeil \cite{VERMEIL} (and independently later by Cartan \cite{CARTAN}). Since then, it has been determinant in several characterizations of electromagnetic energy tensors as explained above, as well as in quantum field theory (\cite{KMCMM,KMM}). 


\section{Preliminaries}

\subsection{Natural tensors}

Let $X$ be a smooth $n$-manifold. Let $\M$ be the sheaf of pseudo-riemannian metrics of fixed signature $(s_+,s_-)$ over $X$, let $\Lambda^2$ be the sheaf of 2-forms over $X$ and let $\Tp$ be the sheaf of $p$-covariant tensors over $X$.

\begin{defin}
A natural tensor $T: \M \times \Lambda^2 \rightarrow \Tp$ (associated to metrics and 2-forms) is a regular\footnote{In this context, regularity is a technical condition that refers to the smoothness of the morphism, and its rigorous definition has been omitted for simplicity. The interested reader can check \cite{KMSBOOK, N_THESIS}.} morphism of sheaves such that, for any local diffeomorphism $\tau:U\rightarrow V$, it holds that
$$ T(\tau^* g,\tau^* F)=\tau^*T(g,F) \, . $$
\end{defin}

\begin{defin}
Let $a,w\in \mathbb{Z}$. We say that a natural tensor $T:\M \times \Lambda^2 \rightarrow \Tp$ is homogeneous of relative weight $(a,w)$ if, for all non-zero $\lambda\in \RR$, it holds that:
$$T(\lambda^2 g, \lambda^a F)=\lambda^w T(g,F) \ .
$$
\end{defin}

\begin{defin} 
Let $m\geq 1$ be an integer. The  space $N_m\,$ of normal tensors of order $m$ at $x$ is the vector subspace of $m+2$-covariant tensors $T$ at $x$ satisfying the
following symmetries:
\begin{enumerate}
\item they are symmetric in the first two and the last $m$ covariant indices:
\begin{equation}
T_{ijk_1\ldots k_m} =T_{jik_1\ldots k_m} \ , \quad  T_{ijk_1\ldots k_m} = T_{ijk_{\sigma(1)} \ldots k_{\sigma(m)}} \ , \quad
\forall \ \sigma \in S_m\,;
\end{equation}

\item the cyclic sum of the last $\,m+1\,$ covariant indices is zero:
\begin{equation}
T_{ijk_1\ldots k_m} + T_{ik_mjk_1\ldots k_{m-1}} + \ldots + T_{ik_1\ldots k_m j} = 0.
\end{equation} 
\end{enumerate}
\end{defin} 

For $m=0$, the space $N_0$ is defined as the set of pseudo-riemannian metrics of the fixed signature $(s_+,s_-)$ at $x$, and for $m=1$ it holds that $N_1=0$ due to the symmetries.

These spaces recover the symmetries of a metric tensor and its partial derivatives at a point in normal coordinates. In the same vein, the symmetries of a $2$-form and its derivatives are given by the spaces $V_r:=\Lambda^2 T_x^* X \otimes S^r T_x^* X$.

\begin{thm}[\cite{N_THESIS}]\label{MainThmMetric}
Let $X$ be a smooth manifold of dimension $n$, let $\M$ denote the sheaves of pseudo-riemannian metrics of fixed signature $(s_+,s_-)$, let $\Lambda^2$ be the sheaf of 2-forms over $X$ and let $\Tp$ be the sheaf of $p$-covariant tensors. Let $a, w \in \mathbb{Z}$.


If we fix a point $x \in X$ and a pseudo-riemannian metric $g_{x}$ of  fixed signature $(s_+,s_-)$ at $x$, there exists a $\mathbb{R}$-linear isomorphism

$$
\begin{CD}
\left\{
\begin{array}{c}
 \,  \text{Natural tensors } \ \\
 \M \times \Lambda^2 \longrightarrow \Tp \ \\
 \text{homogeneous of relative weight } (a,w) \
\end{array} \right\} \\
\Big\Vert \\
\bigoplus  \limits_{d_i, c_j} \Hom{\Ogx}{S^{d_2}N_2 \otimes \ldots \otimes S^{d_r}N_r \otimes S^{c_0}V_0 \otimes \ldots \otimes S^{c_s}V_s}{\Tx} \ ,
\end{CD}
$$
where $d_2, \ldots , d_r$ run over the non-negative integer solutions of the equation
\[
2 d_2 + \ldots + r d_r + (2-a)c_0 + \ldots + (2+s-a)c_s =p-w \ .
\] and where $\Ogx :=  \{  \mathrm{d}_{x} \tau \colon \tau \in \Diff: \tau_{*,x} g_{x}=g_{x} \}$.
\end{thm}

\subsection{Invariant theory of the orthogonal group}

Let $\,V\,$ be a real vector space of finite dimension $\,n$, let $g$ be a non-degenerate symmetric bilinear form on $V$ and let $\,\mathrm{O}(n,\RR) \,$ be the orthogonal group: the real Lie group of  $\mathbb{R}$-linear automorphisms $V \rightarrow V$ that preserve $g$. 

The First Fundamental Theorem of the orthogonal group $\,\mathrm{O}(n,\RR) \,$ allows us to compute $\,\mathrm{O}(n,\RR) \,$-equivariant linear maps:

\begin{thm}[\cite{GW}]\label{MainTheoremO} 
The real vector space $\,\mathrm{Hom}_{\mathrm{O}(n,\RR)}\left( V \otimes \stackrel{p}{\ldots} \otimes V  \, , \, \RR \right) \,$
of invariant linear forms on $\, V \otimes \ldots \otimes V\,$ is null if $p$ is odd, whereas if $p$ is even it is spanned by 
$$ g_\sigma ((e_1 , \ldots , e_p)) := g (e_{\sigma(1)}, e_{\sigma(2)}) \ldots g (e_{\sigma(p-1)}, e_{\sigma(p)}) \ ,   $$
where $\sigma \in S_p .$
\end{thm}

We will also require the following facts:

\begin{propo}[\cite{COLLOQUIUM}]\label{ProposicionInvariantes} Let $E$ and $F$ be (algebraic) linear representations of $\mathrm{O}(n,\RR)$.
\begin{itemize}
\item \label{inv1} There exists a linear isomorphism $\mathrm{Hom}_{\mathrm{O}(n,\RR)} (E , F) = \mathrm{Hom}_{\mathrm{O}(n,\RR)} (E \otimes F^* , \mathbb{R}) $.

\item \label{inv2}
If $W\subseteq E$ is a sub-representation, then any equivariant linear map $W \to F$ is the restriction of an equivariant linear map $E \to F$.
\end{itemize}
\end{propo}

\section{Main Theorem}

For the remainder of this section, let $X$ be a smooth manifold of dimension $n\geq 5$.

\begin{thm}\label{MainTheorem}
There are no non-zero natural 2-tensors $T(g,F)$ associated to metrics and 2-forms that verify the following conditions:
\begin{itemize}
\item[(a)] If $F_x =0$, then $T_x=0$.
\item[(b)] $T$ is homogeneous of relative weight $(1,-2)$.
\item[(c)] $T$ is divergence-free.
\end{itemize}
\end{thm}

\proof
Let us invoke Theorem \ref{MainThmMetric}, which describes all natural $p$-tensors associated to metrics and $2$-forms:
$$
\begin{CD}
\left\{
\begin{aligned}
& \, \text{Natural tensors } \M \times \Lambda^2 \longrightarrow \Tp \ \\
& \text{ homogeneous of weight } (a,w)  \ 
\end{aligned} \right\} 
\\ 
\Big\Vert
\\
\bigoplus \limits_{d_i, c_j} \Hom{\Ogx}{S^{d_2}N_2 \otimes \ldots \otimes S^{d_r}N_r \otimes S^{c_0}V_0 \otimes \ldots \otimes S^{c_s}V_s}{\Tx} \ ,
\end{CD}
$$
where $d_2, \ldots , d_r, c_0, \ldots, c_s$ are non-negative integers running over the solutions of the equation
\[
2 d_2 + \ldots + r d_r + (2-a)c_0 + \ldots + (2+s-a)c_s =p-w \ .
\]

Substituting $a=1$, $p=2$ and $w=-2$ in the equation above leaves us with:

\[
2 d_2 + \ldots + r d_r + c_0 + \ldots + (s+1) c_s = 4 \ .
\]

The only solutions are \footnote{Only non-zero coefficients are shown.} :

\begin{enumerate}
\item $d_2 = 1$, $c_0=2$.

\item $d_2 = 1$, $c_1=1$.

\item $d_2 = 2$.

\item $d_3 = 1$, $c_0=1$.

\item $d_4 = 1$.

\item $c_0=4$.

\item $c_1=2$.

\item $c_0 = 2$, $c_1=1$.

\item $c_0 = 1$, $c_2=1$.

\item $c_3 = 1$.

\end{enumerate}

Out of the 10 solutions, we can immediately discard a few of them: the natural tensors proceeding from any solution with $c_0=0$, due to hypothesis (a). This eliminates solutions (2), (3), (5), (7) and (10)  . 

For the remaining 5 solutions, we analyse the associated $\Ogx$-equivariant maps. All of these maps are equivalent (Proposition \ref{ProposicionInvariantes}) to a $\Ogx$-invariant map of the form

\[
S^{d_2}N_2 \otimes \ldots \otimes S^{d_r}N_r \otimes S^{c_0}V_0 \otimes \ldots \otimes S^{c_s}V_s \otimes T_xX \otimes T_xX  \longrightarrow \RR \ ,
\]
which allows us to use Theorem \ref{MainTheoremO} in order to do the computations.

Directly from the application of the theorem, it follows that the only map associated to solutions (4) and (8) is the null map, as there is an odd amount of indices. Therefore, let us describe the generators of maps associated to solutions (1), (6) and (9):

\begin{itemize}
\item $d_2=1$, $c_0=2$:
As $N_2$ is isomorphic to the space of curvature-like tensors $\R$, we consider $\Ogx$-equivariant maps
\[
\R \otimes S^{2}V_0 \longrightarrow T_x^*X \otimes T_x^*X \ .
\]
Due to Theorem \ref{MainTheoremO}, the generators are every non-zero combination of multiplying by $g$, $g^{-1}$ and contracting indices of the expression
$$R_{ijkl} \ F_{ab} \ F_{cd}$$
resulting in a 2-tensor, where $R_{ijkl}$ denotes the Riemann-Christoffel tensor of $g$. 
Taking into account the symmetries of the expression above, the list of generators is reduced to the following ($\Ric$ denotes the Ricci curvature 2-tensor) :
\begin{itemize}
\item $(T_1)_{ij}=R_{iabc}\ F_j^a \ F^{bc}$.
\item $(T_2)_{ij}=(T_2)_{ji}=R_{jabc}\ F_i^a \ F^{bc}$.
\item $(T_3)_{ij}=R_{abcd} \ F^{ab} \ F^{cd} g_{ij}$.
\item $(T_4)_{ij}=R_{iajb}\ F^{a}_c \ F^{bc}$.
\item $(T_5)_{ij}=F_{ab} \ F^{ab} \ \Ric_{ij}$.
\item $(T_6)_{ij}=\Ric_{ia}\ F_{jb} \ F^{ab}$.
\item $(T_7)_{ij}=(T_6)_{ji}=\Ric_{ja} \ F_{ib} \ F^{ab}$.
\item $(T_8)_{ij}=\Ric_{ab} \ F^{a}_i \ F^{b}_j$.
\item $(T_9)_{ij}=\Ric_{ab} \ F^{a}_c \ F^{cb} \ g_{ij}$.
\item $(T_{10})_{ij}=r \ F_{ia} \ F^{a}_j$.
\item $(T_{11})_{ij}=r \ F_{ab} \ F^{ab}  \ g_{ij}$.
\end{itemize} 

\item $c_0=4$:
Now we look at the $\Ogx$-equivariant maps
\[
 S^{4}V_0 \longrightarrow T_x^*X \otimes T_x^*X \ .
\]

Again, by Theorem \ref{MainTheoremO} we must analyse the index contraction of the expression
$$F_{ij} \ F_{kl} \ F_{ab} \ F_{cd} \ ,$$

obtaining:

\begin{itemize}
\item $(T_{12})_{ij}=F_{ab} \ F^{ab} \ F_{ic} \ F_{j}^c$.
\item $(T_{13})_{ij}=F_{c}^a  \ F^{bc} \ F_{ia} \ F_{jb}$.
\item $(T_{14})_{ij}=F_{ab} \ F^{ab} \ F_{cd} \ F^{cd} \ g_{ij}$.
\item $(T_{15})_{ij}=F_{ab} \ F^{bc} \ F_{cd} \ F^{da} \ g_{ij}$.
\end{itemize}

\item $c_0=1$, $c_2=1$:
Lastly, we check the $\Ogx$-equivariant maps
\[
 V_0 \otimes V_2 \longrightarrow T_x^*X \otimes T_x^*X \ ,
\]
corresponding to index contraction of the expression
$$F_{ij,kl} \ F_{ab} \ ,$$
where $F_{ij,kl}$ denotes the coordinates of the normal tensor of order 2 of $F$.

The possibilities now go as follows:

\begin{itemize}
\item $(T_{16})_{ij}=F_{ia,jb} \ F^{ab} $.
\item $(T_{17})_{ij}=F_{ab,ij} \ F^{ab} $.
\item $(T_{18})_{ij}=F_{ia,b}^{\ \ \ \, b} \ F_j^{\ a} $.
\item $(T_{19})_{ij}=F_{ia,b}^{\ \ \ \, a} \ F_j^{\ b} $.
\item $(T_{20})_{ij}=F_{ab,c}^{\ \ \ \, c} \ F^{a b} g_{ij} $.
\item $(T_{21})_{ij}=F_{ab,c}^{\ \ \ \, b} \ F^{a c} g_{ij} $.
\end{itemize} 

\end{itemize}

To summarize, the 21 tensors listed above generate every natural 2-tensor associated to metrics and 2-forms verifying conditions (a), (b) and (c). Therefore, in order to prove the thesis stated in the theorem, we must check that there are no divergence-free 2-tensors generated by the ones on the list. Equivalently, it is enough to prove that their divergences are linearly independent.

As we are working with natural tensors and their divergences are also natural tensors, if we find an example of a manifold equipped with a metric and a 2-form in which the 21 divergences are linearly independent, then they automatically are linearly independent as natural tensors. 

The following example checks out: for $n= 5$,  consider the manifold 
$$
X := \{ (x_1, \ldots, x_5) \in \RR^5 : x_2>0, \ x_3>0, \ x_4>0 \}
$$
equipped with the metric

$$
g=\dxdx{1}{1}+\dxdx{2}{2}+\dxdx{3}{3}+x_{2}\,x_{3}\,x_{4}\,(\dxdx{4}{4}+\,\dxdx{5}{5})
$$

and the 2-form

$$
F=\left({x_{1}}^2\,{x_{3}}^2\right)\,\d x_{1}\,\wedge \, \d x_{2}+\d x_{3}\,\wedge \, \d x_{4}
$$

In this setting, the 21 tensors have the following expression:

\begin{dgroup*}
\begin{dmath*}T_{1}=-\frac{1}{2\,{x_{2}}^2\,{x_{3}}^2\,x_{4}}\,\dxdx{2}{3}+\frac{1}{2\,x_{2}\,{x_{3}}^3\,x_{4}}\,\dxdx{3}{3}+\frac{{x_{1}}^2\,x_{3}}{2\,x_{2}}\,\dxdx{4}{1}+\frac{1}{2\,{x_{3}}^2}\,\dxdx{4}{4} \ , \end{dmath*}
     \begin{dmath*}T_{2}=\frac{{x_{1}}^2\,x_{3}}{2\,x_{2}}\,\dxdx{1}{4}-\frac{1}{2\,{x_{2}}^2\,{x_{3}}^2\,x_{4}}\,\dxdx{3}{2}+\frac{1}{2\,x_{2}\,{x_{3}}^3\,x_{4}}\,\dxdx{3}{3}+\frac{1}{2\,{x_{3}}^2}\,\dxdx{4}{4} \ , \end{dmath*}
     \begin{dmath*}T_{3}=\frac{1}{x_{2}\,{x_{3}}^3\,x_{4}}\,\dxdx{1}{1}+\frac{1}{x_{2}\,{x_{3}}^3\,x_{4}}\,\dxdx{2}{2} + \frac{1}{x_{2}\,{x_{3}}^3\,x_{4}}\,\dxdx{3}{3}+\frac{\dxdx{4}{4}}{{x_{3}}^2}+\frac{\dxdx{5}{5}}{{x_{3}}^2} \ , \end{dmath*}
     \begin{dmath*}T_{4}=\frac{1}{4\,{x_{2}}^3\,x_{3}\,x_{4}}\,\dxdx{2}{2}-\frac{1}{4\,{x_{2}}^2\,{x_{3}}^2\,x_{4}}\,\dxdx{2}{3}-\frac{1}{4\,{x_{2}}^2\,{x_{3}}^2\,x_{4}}\,\dxdx{3}{2}+\frac{1}{4\,x_{2}\,{x_{3}}^3\,x_{4}}\,\dxdx{3}{3}+\frac{x_{4}\,{x_{1}}^4\,{x_{3}}^7+x_{2}}{4\,x_{2}\,{x_{3}}^2}\,\dxdx{4}{4}+\frac{x_{2}\,{x_{1}}^4\,{x_{3}}^6\,{x_{4}}^4-x_{3}\,{x_{4}}^3+2\,x_{2}}{4\,{x_{2}}^2\,x_{3}\,{x_{4}}^3}\,\dxdx{5}{5} \ , \end{dmath*}
     \begin{dmath*}T_{5}=-\frac{x_{2}\,x_{4}\,{x_{1}}^4\,{x_{3}}^5+1}{{x_{2}}^3\,x_{3}\,x_{4}}\,\dxdx{2}{2}+\frac{x_{2}\,x_{4}\,{x_{1}}^4\,{x_{3}}^5+1}{{x_{2}}^2\,{x_{3}}^2\,x_{4}}\,\dxdx{2}{3}+\frac{x_{2}\,x_{4}\,{x_{1}}^4\,{x_{3}}^5+1}{{x_{2}}^2\,{x_{3}}^2\,x_{4}}\,\dxdx{3}{2}-\frac{x_{2}\,x_{4}\,{x_{1}}^4\,{x_{3}}^5+1}{x_{2}\,{x_{3}}^3\,x_{4}}\,\dxdx{3}{3}-\frac{x_{2}\,x_{4}\,{x_{1}}^4\,{x_{3}}^5+1}{x_{2}\,x_{3}\,{x_{4}}^3}\,\dxdx{4}{4}-\frac{x_{2}\,x_{4}\,{x_{1}}^4\,{x_{3}}^5+1}{x_{2}\,x_{3}\,{x_{4}}^3}\,\dxdx{5}{5} \ , \end{dmath*}
     \begin{dmath*}T_{6}=-\frac{{x_{1}}^4\,{x_{3}}^4}{2\,{x_{2}}^2}\,\dxdx{2}{2}+\frac{1}{2\,{x_{2}}^2\,{x_{3}}^2\,x_{4}}\,\dxdx{2}{3}+\frac{{x_{1}}^4\,{x_{3}}^3}{2\,x_{2}}\,\dxdx{3}{2}-\frac{1}{2\,x_{2}\,{x_{3}}^3\,x_{4}}\,\dxdx{3}{3}-\frac{1}{2\,x_{2}\,x_{3}\,{x_{4}}^3}\,\dxdx{4}{4} \ , \end{dmath*}
     \begin{dmath*}T_{7}=-\frac{{x_{1}}^4\,{x_{3}}^4}{2\,{x_{2}}^2}\,\dxdx{2}{2}+\frac{{x_{1}}^4\,{x_{3}}^3}{2\,x_{2}}\,\dxdx{2}{3}+\frac{1}{2\,{x_{2}}^2\,{x_{3}}^2\,x_{4}}\,\dxdx{3}{2}-\frac{1}{2\,x_{2}\,{x_{3}}^3\,x_{4}}\,\dxdx{3}{3}-\frac{1}{2\,x_{2}\,x_{3}\,{x_{4}}^3}\,\dxdx{4}{4} \ , \end{dmath*}
     \begin{dmath*}T_{8}=-\frac{{x_{1}}^4\,{x_{3}}^4}{2\,{x_{2}}^2}\,\dxdx{1}{1}-\frac{{x_{1}}^2\,x_{3}}{2\,x_{2}}\,\dxdx{1}{4}-\frac{1}{2\,{x_{2}}^2\,{x_{3}}^2\,{x_{4}}^4}\,\dxdx{3}{3}-\frac{{x_{1}}^2\,x_{3}}{2\,x_{2}}\,\dxdx{4}{1}-\frac{1}{2\,{x_{3}}^2}\,\dxdx{4}{4} \ , \end{dmath*}
     \begin{dmath*}T_{9}=-\frac{{x_{1}}^4\,{x_{3}}^7\,{x_{4}}^4+x_{3}+x_{2}\,{x_{4}}^3}{2\,{x_{2}}^2\,{x_{3}}^3\,{x_{4}}^4}\,\dxdx{1}{1}-\frac{{x_{1}}^4\,{x_{3}}^7\,{x_{4}}^4+x_{3}+x_{2}\,{x_{4}}^3}{2\,{x_{2}}^2\,{x_{3}}^3\,{x_{4}}^4}\,\dxdx{2}{2}-\frac{{x_{1}}^4\,{x_{3}}^7\,{x_{4}}^4+x_{3}+x_{2}\,{x_{4}}^3}{2\,{x_{2}}^2\,{x_{3}}^3\,{x_{4}}^4}\,\dxdx{3}{3}-\frac{{x_{1}}^4\,{x_{3}}^7\,{x_{4}}^4+x_{3}+x_{2}\,{x_{4}}^3}{2\,x_{2}\,{x_{3}}^2\,{x_{4}}^3}\,\dxdx{4}{4}-\frac{{x_{1}}^4\,{x_{3}}^7\,{x_{4}}^4+x_{3}+x_{2}\,{x_{4}}^3}{2\,x_{2}\,{x_{3}}^2\,{x_{4}}^3}\,\dxdx{5}{5} \ , \end{dmath*}
     \begin{dmath*}T_{10}=-\frac{{x_{1}}^4\,{x_{3}}^2\,\left({x_{2}}^2\,{x_{4}}^3+2\,x_{2}\,x_{3}+{x_{3}}^2\,{x_{4}}^3\right)}{2\,{x_{2}}^2\,{x_{4}}^3}\,\dxdx{1}{1}-\frac{{x_{1}}^4\,{x_{3}}^2\,\left({x_{2}}^2\,{x_{4}}^3+2\,x_{2}\,x_{3}+{x_{3}}^2\,{x_{4}}^3\right)}{2\,{x_{2}}^2\,{x_{4}}^3}\,\dxdx{2}{2}-\frac{{x_{2}}^2\,{x_{4}}^3+2\,x_{2}\,x_{3}+{x_{3}}^2\,{x_{4}}^3}{2\,{x_{2}}^3\,{x_{3}}^3\,{x_{4}}^4}\,\dxdx{3}{3}-\frac{{x_{2}}^2\,{x_{4}}^3+2\,x_{2}\,x_{3}+{x_{3}}^2\,{x_{4}}^3}{2\,{x_{2}}^2\,{x_{3}}^2\,{x_{4}}^3}\,\dxdx{4}{4} \ , \end{dmath*}
     \begin{dmath*}T_{11}=-\frac{\left(x_{2}\,x_{4}\,{x_{1}}^4\,{x_{3}}^5+1\right)\,\left({x_{2}}^2\,{x_{4}}^3+2\,x_{2}\,x_{3}+{x_{3}}^2\,{x_{4}}^3\right)}{{x_{2}}^3\,{x_{3}}^3\,{x_{4}}^4}\,\dxdx{1}{1}-\frac{\left(x_{2}\,x_{4}\,{x_{1}}^4\,{x_{3}}^5+1\right)\,\left({x_{2}}^2\,{x_{4}}^3+2\,x_{2}\,x_{3}+{x_{3}}^2\,{x_{4}}^3\right)}{{x_{2}}^3\,{x_{3}}^3\,{x_{4}}^4}\,\dxdx{2}{2}-\frac{\left(x_{2}\,x_{4}\,{x_{1}}^4\,{x_{3}}^5+1\right)\,\left({x_{2}}^2\,{x_{4}}^3+2\,x_{2}\,x_{3}+{x_{3}}^2\,{x_{4}}^3\right)}{{x_{2}}^3\,{x_{3}}^3\,{x_{4}}^4}\,\dxdx{3}{3}-\frac{\left(x_{2}\,x_{4}\,{x_{1}}^4\,{x_{3}}^5+1\right)\,\left({x_{2}}^2\,{x_{4}}^3+2\,x_{2}\,x_{3}+{x_{3}}^2\,{x_{4}}^3\right)}{{x_{2}}^2\,{x_{3}}^2\,{x_{4}}^3}\,\dxdx{4}{4}-\frac{\left(x_{2}\,x_{4}\,{x_{1}}^4\,{x_{3}}^5+1\right)\,\left({x_{2}}^2\,{x_{4}}^3+2\,x_{2}\,x_{3}+{x_{3}}^2\,{x_{4}}^3\right)}{{x_{2}}^2\,{x_{3}}^2\,{x_{4}}^3}\,\dxdx{5}{5} \ , \end{dmath*}
     \begin{dmath*}T_{12}=\frac{2\,{x_{1}}^4\,{x_{3}}^3\,\left(x_{2}\,x_{4}\,{x_{1}}^4\,{x_{3}}^5+1\right)}{x_{2}\,x_{4}}\,\dxdx{1}{1}+\frac{2\,{x_{1}}^4\,{x_{3}}^3\,\left(x_{2}\,x_{4}\,{x_{1}}^4\,{x_{3}}^5+1\right)}{x_{2}\,x_{4}}\,\dxdx{2}{2}+\frac{2\,x_{2}\,x_{4}\,{x_{1}}^4\,{x_{3}}^5+2}{{x_{2}}^2\,{x_{3}}^2\,{x_{4}}^2}\,\dxdx{3}{3}+\frac{2\,\left(x_{2}\,x_{4}\,{x_{1}}^4\,{x_{3}}^5+1\right)}{x_{2}\,x_{3}\,x_{4}}\,\dxdx{4}{4} \ , \end{dmath*}
     \begin{dmath*}T_{13}=\left(-{x_{1}}^8\,{x_{3}}^8\right)\,\dxdx{1}{1}+\left(-{x_{1}}^8\,{x_{3}}^8\right)\,\dxdx{2}{2}-\frac{1}{{x_{2}}^2\,{x_{3}}^2\,{x_{4}}^2}\,\dxdx{3}{3}-\frac{1}{x_{2}\,x_{3}\,x_{4}}\,\dxdx{4}{4} \ , \end{dmath*}
     \begin{dmath*}T_{14}=\frac{4\,{\left(x_{2}\,x_{4}\,{x_{1}}^4\,{x_{3}}^5+1\right)}^2}{{x_{2}}^2\,{x_{3}}^2\,{x_{4}}^2}\,\dxdx{1}{1}+\frac{4\,{\left(x_{2}\,x_{4}\,{x_{1}}^4\,{x_{3}}^5+1\right)}^2}{{x_{2}}^2\,{x_{3}}^2\,{x_{4}}^2}\,\dxdx{2}{2}+\frac{4\,{\left(x_{2}\,x_{4}\,{x_{1}}^4\,{x_{3}}^5+1\right)}^2}{{x_{2}}^2\,{x_{3}}^2\,{x_{4}}^2}\,\dxdx{3}{3}+\frac{4\,{\left(x_{2}\,x_{4}\,{x_{1}}^4\,{x_{3}}^5+1\right)}^2}{x_{2}\,x_{3}\,x_{4}}\,\dxdx{4}{4}+\frac{4\,{\left(x_{2}\,x_{4}\,{x_{1}}^4\,{x_{3}}^5+1\right)}^2}{x_{2}\,x_{3}\,x_{4}}\,\dxdx{5}{5} \ , \end{dmath*}
     \begin{dmath*}T_{15}=\frac{2\,\left({x_{1}}^8\,{x_{2}}^2\,{x_{3}}^{10}\,{x_{4}}^2+1\right)}{{x_{2}}^2\,{x_{3}}^2\,{x_{4}}^2}\,\dxdx{1}{1}+\frac{2\,\left({x_{1}}^8\,{x_{2}}^2\,{x_{3}}^{10}\,{x_{4}}^2+1\right)}{{x_{2}}^2\,{x_{3}}^2\,{x_{4}}^2}\,\dxdx{2}{2}+\frac{2\,\left({x_{1}}^8\,{x_{2}}^2\,{x_{3}}^{10}\,{x_{4}}^2+1\right)}{{x_{2}}^2\,{x_{3}}^2\,{x_{4}}^2}\,\dxdx{3}{3}+\frac{2\,\left({x_{1}}^8\,{x_{2}}^2\,{x_{3}}^{10}\,{x_{4}}^2+1\right)}{x_{2}\,x_{3}\,x_{4}}\,\dxdx{4}{4}+\frac{2\,\left({x_{1}}^8\,{x_{2}}^2\,{x_{3}}^{10}\,{x_{4}}^2+1\right)}{x_{2}\,x_{3}\,x_{4}}\,\dxdx{5}{5} \ , \end{dmath*}
     \begin{dmath*}T_{16}=\left(-2\,{x_{1}}^2\,{x_{3}}^4\right)\,\dxdx{1}{1}+\left(-4\,{x_{1}}^3\,{x_{3}}^3\right)\,\dxdx{1}{3} \ , \end{dmath*}
     \begin{dmath*}T_{17}=\left(4\,{x_{1}}^2\,{x_{3}}^4\right)\,\dxdx{1}{1}+\left(8\,{x_{1}}^3\,{x_{3}}^3\right)\,\dxdx{1}{3}+\left(8\,{x_{1}}^3\,{x_{3}}^3\right)\,\dxdx{3}{1}+\left(4\,{x_{1}}^4\,{x_{3}}^2\right)\,\dxdx{3}{3} \ , \end{dmath*}
     \begin{dmath*}T_{18}=\left(2\,{x_{1}}^2\,{x_{3}}^2\,\left({x_{1}}^2+{x_{3}}^2\right)\right)\,\dxdx{1}{1}+\left(2\,{x_{1}}^2\,{x_{3}}^2\,\left({x_{1}}^2+{x_{3}}^2\right)\right)\,\dxdx{2}{2} \ , \end{dmath*}
     \begin{dmath*}T_{19}=\left(2\,{x_{1}}^2\,{x_{3}}^4\right)\,\dxdx{2}{2}+\left(4\,x_{1}\,x_{3}\right)\,\dxdx{2}{4} \ , \end{dmath*}
     \begin{dmath*}T_{20}=\left(4\,{x_{1}}^2\,{x_{3}}^2\,\left({x_{1}}^2+{x_{3}}^2\right)\right)\,\dxdx{1}{1}+\left(4\,{x_{1}}^2\,{x_{3}}^2\,\left({x_{1}}^2+{x_{3}}^2\right)\right)\,\dxdx{2}{2}+\left(4\,{x_{1}}^2\,{x_{3}}^2\,\left({x_{1}}^2+{x_{3}}^2\right)\right)\,\dxdx{3}{3}+\left(4\,{x_{1}}^2\,x_{2}\,{x_{3}}^3\,x_{4}\,\left({x_{1}}^2+{x_{3}}^2\right)\right)\,\dxdx{4}{4}+\left(4\,{x_{1}}^2\,x_{2}\,{x_{3}}^3\,x_{4}\,\left({x_{1}}^2+{x_{3}}^2\right)\right)\,\dxdx{5}{5} \ , \end{dmath*}
     \begin{dmath*}T_{21}=\left(2\,{x_{1}}^2\,{x_{3}}^4\right)\,\dxdx{1}{1}+\left(2\,{x_{1}}^2\,{x_{3}}^4\right)\,\dxdx{2}{2}+\left(2\,{x_{1}}^2\,{x_{3}}^4\right)\,\dxdx{3}{3}+\left(2\,{x_{1}}^2\,x_{2}\,{x_{3}}^5\,x_{4}\right)\,\dxdx{4}{4}+\left(2\,{x_{1}}^2\,x_{2}\,{x_{3}}^5\,x_{4}\right)\,\dxdx{5}{5} \ . \end{dmath*}
     \end{dgroup*}

\vspace*{1cm}

Their divergences are expressed, in the same coordinates, as follows:

\vspace*{1cm}

\begin{dgroup*}
\begin{dmath*}\text{div } T_{1} =-\frac{1}{4\,{x_{2}}^2\,{x_{3}}^3\,x_{4}}\,\d x_{2}-\frac{5\,{x_{2}}^2-2\,{x_{3}}^2}{4\,{x_{2}}^3\,{x_{3}}^4\,x_{4}}\,\d x_{3}-\frac{1}{4\,x_{2}\,{x_{3}}^3\,{x_{4}}^2}\,\d x_{4} \ , \end{dmath*}
     \begin{dmath*}\text{div } T_{2} =\frac{1}{4\,{x_{2}}^2\,{x_{3}}^3\,x_{4}}\,\d x_{2}-\frac{5}{4\,x_{2}\,{x_{3}}^4\,x_{4}}\,\d x_{3}+\frac{4\,x_{1}\,{x_{3}}^4\,{x_{4}}^2-1}{4\,x_{2}\,{x_{3}}^3\,{x_{4}}^2}\,\d x_{4} \ , \end{dmath*}
     \begin{dmath*}\text{div } T_{3} =-\frac{1}{{x_{2}}^2\,{x_{3}}^3\,x_{4}}\,\d x_{2}-\frac{3}{x_{2}\,{x_{3}}^4\,x_{4}}\,\d x_{3}-\frac{1}{x_{2}\,{x_{3}}^3\,{x_{4}}^2}\,\d x_{4} \ , \end{dmath*}
     \begin{dmath*}\text{div } T_{4} =-\frac{2\,{x_{1}}^4\,x_{2}\,{x_{3}}^7\,{x_{4}}^4-{x_{2}}^2\,{x_{4}}^3+2\,x_{2}\,x_{3}+3\,{x_{3}}^2\,{x_{4}}^3}{8\,{x_{2}}^4\,{x_{3}}^3\,{x_{4}}^4}\,\d x_{2}-\frac{2\,{x_{1}}^4\,x_{2}\,{x_{3}}^7\,{x_{4}}^4+5\,{x_{2}}^2\,{x_{4}}^3+2\,x_{2}\,x_{3}-3\,{x_{3}}^2\,{x_{4}}^3}{8\,{x_{2}}^3\,{x_{3}}^4\,{x_{4}}^4}\,\d x_{3}-\frac{{x_{2}}^2\,{x_{4}}^3+2\,x_{2}\,x_{3}-{x_{3}}^2\,{x_{4}}^3}{8\,{x_{2}}^3\,{x_{3}}^3\,{x_{4}}^5}\,\d x_{4} \ , \end{dmath*}
     \begin{dmath*}\text{div } T_{5} =\frac{4\,{x_{1}}^4\,{x_{2}}^3\,{x_{3}}^5\,{x_{4}}^4+{x_{1}}^4\,{x_{2}}^2\,{x_{3}}^6\,x_{4}+{x_{1}}^4\,x_{2}\,{x_{3}}^7\,{x_{4}}^4-{x_{2}}^2\,{x_{4}}^3+x_{2}\,x_{3}+2\,{x_{3}}^2\,{x_{4}}^3}{{x_{2}}^4\,{x_{3}}^3\,{x_{4}}^4}\,\d x_{2}+\frac{-3\,{x_{1}}^4\,{x_{2}}^3\,{x_{3}}^5\,{x_{4}}^4+{x_{1}}^4\,{x_{2}}^2\,{x_{3}}^6\,x_{4}+2\,{x_{2}}^2\,{x_{4}}^3+x_{2}\,x_{3}-{x_{3}}^2\,{x_{4}}^3}{{x_{2}}^3\,{x_{3}}^4\,{x_{4}}^4}\,\d x_{3}+\frac{3\,{x_{1}}^4\,x_{2}\,{x_{3}}^5\,x_{4}+4}{{x_{2}}^2\,{x_{3}}^2\,{x_{4}}^5}\,\d x_{4} \ , \end{dmath*}
     \begin{dmath*}\text{div } T_{6} =\frac{8\,{x_{1}}^4\,{x_{2}}^2\,{x_{3}}^4\,{x_{4}}^4+2\,{x_{1}}^4\,{x_{3}}^6\,{x_{4}}^4+1}{4\,{x_{2}}^3\,{x_{3}}^2\,{x_{4}}^4}\,\d x_{2}+\frac{4\,{x_{2}}^2\,{x_{4}}^3+x_{2}\,x_{3}-2\,{x_{3}}^2\,{x_{4}}^3}{4\,{x_{2}}^3\,{x_{3}}^4\,{x_{4}}^4}\,\d x_{3}+\frac{7}{4\,{x_{2}}^2\,{x_{3}}^2\,{x_{4}}^5}\,\d x_{4} \ , \end{dmath*}
     \begin{dmath*}\text{div } T_{7} =\frac{2\,{x_{1}}^4\,{x_{3}}^7\,{x_{4}}^4-2\,x_{2}\,{x_{4}}^3+x_{3}}{4\,{x_{2}}^3\,{x_{3}}^3\,{x_{4}}^4}\,\d x_{2}+\frac{4\,x_{2}\,{x_{4}}^3+x_{3}}{4\,{x_{2}}^2\,{x_{3}}^4\,{x_{4}}^4}\,\d x_{3}+\frac{7}{4\,{x_{2}}^2\,{x_{3}}^2\,{x_{4}}^5}\,\d x_{4} \ , \end{dmath*}
     \begin{dmath*}\text{div } T_{8} =-\frac{2\,{x_{1}}^3\,{x_{3}}^4}{{x_{2}}^2}\,\d x_{1}+\frac{1}{4\,{x_{2}}^2\,{x_{3}}^3\,x_{4}}\,\d x_{2}+\frac{x_{2}\,{x_{4}}^3+2\,x_{3}}{4\,{x_{2}}^2\,{x_{3}}^4\,{x_{4}}^4}\,\d x_{3}-\frac{4\,x_{1}\,{x_{3}}^4\,{x_{4}}^2-1}{4\,x_{2}\,{x_{3}}^3\,{x_{4}}^2}\,\d x_{4} \ , \end{dmath*}
     \begin{dmath*}\text{div } T_{9} =-\frac{2\,{x_{1}}^3\,{x_{3}}^4}{{x_{2}}^2}\,\d x_{1}+\frac{2\,{x_{1}}^4\,{x_{3}}^7\,{x_{4}}^4+x_{2}\,{x_{4}}^3+2\,x_{3}}{2\,{x_{2}}^3\,{x_{3}}^3\,{x_{4}}^4}\,\d x_{2}+\frac{-4\,{x_{1}}^4\,{x_{3}}^7\,{x_{4}}^4+3\,x_{2}\,{x_{4}}^3+2\,x_{3}}{2\,{x_{2}}^2\,{x_{3}}^4\,{x_{4}}^4}\,\d x_{3}+\frac{x_{2}\,{x_{4}}^3+4\,x_{3}}{2\,{x_{2}}^2\,{x_{3}}^3\,{x_{4}}^5}\,\d x_{4} \ , \end{dmath*}
     \begin{dmath*}\text{div } T_{10} =-\frac{2\,{x_{1}}^3\,{x_{3}}^2\,\left({x_{2}}^2\,{x_{4}}^3+2\,x_{2}\,x_{3}+{x_{3}}^2\,{x_{4}}^3\right)}{{x_{2}}^2\,{x_{4}}^3}\,\d x_{1}+\frac{-2\,{x_{1}}^4\,{x_{2}}^3\,{x_{3}}^5\,{x_{4}}^4+2\,{x_{1}}^4\,x_{2}\,{x_{3}}^7\,{x_{4}}^4+{x_{2}}^2\,{x_{4}}^3+2\,x_{2}\,x_{3}+{x_{3}}^2\,{x_{4}}^3}{4\,{x_{2}}^4\,{x_{3}}^3\,{x_{4}}^4}\,\d x_{2}+\frac{5\,{x_{2}}^2\,{x_{4}}^3+6\,x_{2}\,x_{3}+{x_{3}}^2\,{x_{4}}^3}{4\,{x_{2}}^3\,{x_{3}}^4\,{x_{4}}^4}\,\d x_{3}+\frac{{x_{2}}^2\,{x_{4}}^3+14\,x_{2}\,x_{3}+{x_{3}}^2\,{x_{4}}^3}{4\,{x_{2}}^3\,{x_{3}}^3\,{x_{4}}^5}\,\d x_{4} \ , \end{dmath*}
     \begin{dmath*}\text{div } T_{11} =-\frac{4\,{x_{1}}^3\,{x_{3}}^2\,\left({x_{2}}^2\,{x_{4}}^3+2\,x_{2}\,x_{3}+{x_{3}}^2\,{x_{4}}^3\right)}{{x_{2}}^2\,{x_{4}}^3}\,\d x_{1}+\frac{2\,{x_{1}}^4\,{x_{2}}^2\,{x_{3}}^6\,x_{4}+2\,{x_{1}}^4\,x_{2}\,{x_{3}}^7\,{x_{4}}^4+{x_{2}}^2\,{x_{4}}^3+4\,x_{2}\,x_{3}+3\,{x_{3}}^2\,{x_{4}}^3}{{x_{2}}^4\,{x_{3}}^3\,{x_{4}}^4}\,\d x_{2}+\frac{-2\,{x_{1}}^4\,{x_{2}}^3\,{x_{3}}^5\,{x_{4}}^4-6\,{x_{1}}^4\,{x_{2}}^2\,{x_{3}}^6\,x_{4}-4\,{x_{1}}^4\,x_{2}\,{x_{3}}^7\,{x_{4}}^4+3\,{x_{2}}^2\,{x_{4}}^3+4\,x_{2}\,x_{3}+{x_{3}}^2\,{x_{4}}^3}{{x_{2}}^3\,{x_{3}}^4\,{x_{4}}^4}\,\d x_{3}+\frac{6\,{x_{1}}^4\,{x_{2}}^2\,{x_{3}}^6\,x_{4}+{x_{2}}^2\,{x_{4}}^3+8\,x_{2}\,x_{3}+{x_{3}}^2\,{x_{4}}^3}{{x_{2}}^3\,{x_{3}}^3\,{x_{4}}^5}\,\d x_{4} \ , \end{dmath*}
     \begin{dmath*}\text{div } T_{12} =\frac{8\,{x_{1}}^3\,{x_{3}}^3\,\left(2\,{x_{1}}^4\,x_{2}\,{x_{3}}^5\,x_{4}+1\right)}{x_{2}\,x_{4}}\,\d x_{1}-\frac{-2\,{x_{1}}^8\,{x_{2}}^2\,{x_{3}}^{10}\,{x_{4}}^2+{x_{1}}^4\,x_{2}\,{x_{3}}^5\,x_{4}+1}{{x_{2}}^3\,{x_{3}}^2\,{x_{4}}^2}\,\d x_{2}+\frac{7\,{x_{1}}^4\,x_{2}\,{x_{3}}^5\,x_{4}-3}{{x_{2}}^2\,{x_{3}}^3\,{x_{4}}^2}\,\d x_{3}-\frac{{x_{1}}^4\,x_{2}\,{x_{3}}^5\,x_{4}+3}{{x_{2}}^2\,{x_{3}}^2\,{x_{4}}^3}\,\d x_{4} \ , \end{dmath*}
     \begin{dmath*}\text{div } T_{13} =\left(-8\,{x_{1}}^7\,{x_{3}}^8\right)\,\d x_{1}-\frac{2\,{x_{1}}^8\,{x_{2}}^2\,{x_{3}}^{10}\,{x_{4}}^2-1}{2\,{x_{2}}^3\,{x_{3}}^2\,{x_{4}}^2}\,\d x_{2}+\frac{3}{2\,{x_{2}}^2\,{x_{3}}^3\,{x_{4}}^2}\,\d x_{3}+\frac{3}{2\,{x_{2}}^2\,{x_{3}}^2\,{x_{4}}^3}\,\d x_{4} \ , \end{dmath*}
     \begin{dmath*}\text{div } T_{14} =\frac{32\,{x_{1}}^3\,{x_{3}}^3\,\left({x_{1}}^4\,x_{2}\,{x_{3}}^5\,x_{4}+1\right)}{x_{2}\,x_{4}}\,\d x_{1}-\frac{8\,\left({x_{1}}^4\,x_{2}\,{x_{3}}^5\,x_{4}+1\right)}{{x_{2}}^3\,{x_{3}}^2\,{x_{4}}^2}\,\d x_{2}+\frac{8\,\left(4\,{x_{1}}^8\,{x_{2}}^2\,{x_{3}}^{10}\,{x_{4}}^2+3\,{x_{1}}^4\,x_{2}\,{x_{3}}^5\,x_{4}-1\right)}{{x_{2}}^2\,{x_{3}}^3\,{x_{4}}^2}\,\d x_{3}-\frac{8\,\left({x_{1}}^4\,x_{2}\,{x_{3}}^5\,x_{4}+1\right)}{{x_{2}}^2\,{x_{3}}^2\,{x_{4}}^3}\,\d x_{4} \ , \end{dmath*}
     \begin{dmath*}\text{div } T_{15} =\left(16\,{x_{1}}^7\,{x_{3}}^8\right)\,\d x_{1}-\frac{4}{{x_{2}}^3\,{x_{3}}^2\,{x_{4}}^2}\,\d x_{2}+\frac{4\,\left(4\,{x_{1}}^8\,{x_{2}}^2\,{x_{3}}^{10}\,{x_{4}}^2-1\right)}{{x_{2}}^2\,{x_{3}}^3\,{x_{4}}^2}\,\d x_{3}-\frac{4}{{x_{2}}^2\,{x_{3}}^2\,{x_{4}}^3}\,\d x_{4} \ , \end{dmath*}
     \begin{dmath*}\text{div } T_{16} =\left(-4\,x_{1}\,{x_{3}}^4\right)\,\d x_{1}+\left(-12\,{x_{1}}^2\,{x_{3}}^3\right)\,\d x_{3} \ , \end{dmath*}
     \begin{dmath*}\text{div } T_{17} =\left(8\,x_{1}\,{x_{3}}^2\,\left(4\,{x_{1}}^2+{x_{3}}^2\right)\right)\,\d x_{1}+\left(12\,{x_{1}}^2\,x_{3}\,\left({x_{1}}^2+2\,{x_{3}}^2\right)\right)\,\d x_{3} \ , \end{dmath*}
     \begin{dmath*}\text{div } T_{18} =\left(4\,x_{1}\,{x_{3}}^2\,\left(2\,{x_{1}}^2+{x_{3}}^2\right)\right)\,\d x_{1}+\frac{2\,{x_{1}}^2\,{x_{3}}^2\,\left({x_{1}}^2+{x_{3}}^2\right)}{x_{2}}\,\d x_{2} \ , \end{dmath*}
     \begin{dmath*}\text{div } T_{19} =\frac{2\,{x_{1}}^2\,{x_{3}}^4}{x_{2}}\,\d x_{2}+\frac{2\,x_{1}\,x_{3}}{x_{2}}\,\d x_{4} \ , \end{dmath*}
     \begin{dmath*}\text{div } T_{20} =\left(8\,x_{1}\,{x_{3}}^2\,\left(2\,{x_{1}}^2+{x_{3}}^2\right)\right)\,\d x_{1}+\left(8\,{x_{1}}^2\,x_{3}\,\left({x_{1}}^2+2\,{x_{3}}^2\right)\right)\,\d x_{3} \ , \end{dmath*}
     \begin{dmath*}\text{div } T_{21} =\left(4\,x_{1}\,{x_{3}}^4\right)\,\d x_{1}+\left(8\,{x_{1}}^2\,{x_{3}}^3\right)\,\d x_{3} \ . \end{dmath*}
     \end{dgroup*}

It can be checked that these divergences are $\RR$-linearly independent, finishing the proof.

\qed

\begin{cor}
There exists no natural 2-tensor $T$ associated to metrics and 2-forms that verifies the following properties:
\begin{enumerate}
\item[(a)] If $F_x=0$, then $T_x=0$.
\item[(b)] $T$ is homogeneous of relative weight $(1,-2)$
\item[(c)] $\div T = -i_{\partial F} F$.
\end{enumerate}
\end{cor}
\begin{proof}
Let $E$ be the Maxwell tensor, as defined in Equation \ref{EMT}. Recall that $\div E= -i_{\partial F} F$, and so by the same arguments as above it suffices to find an example of a manifold equipped with a metric and a 2-form in which $\div E$ and the divergences of the 21 tensors found in the proof of Theorem \ref{MainTheorem} are $\RR$-linearly independent. The same example works: defining $X$, $g$ and $F$ as in the previous proof, the tensor $E$ and its divergence are expressed as
\begin{dgroup*}
\begin{dmath*} E=-\frac{1+3\,{x_{1}}^4\,x_{2}\,{x_{3}}^5\,x_{4}}{2\,x_{2}\,x_{3}\,x_{4}}\,\dxdx{1}{1}-\frac{1+3\,{x_{1}}^4\,x_{2}\,{x_{3}}^5\,x_{4}}{2\,x_{2}\,x_{3}\,x_{4}}\,\dxdx{2}{2}-\frac{3+{x_{1}}^4\,x_{2}\,{x_{3}}^5\,x_{4}}{2\,x_{2}\,x_{3}\,x_{4}}\,\dxdx{3}{3}-\frac{3+{x_{1}}^4\,x_{2}\,{x_{3}}^5\,x_{4}}{2}\,\dxdx{4}{4}-\frac{1+{x_{1}}^4\,x_{2}\,{x_{3}}^5\,x_{4}}{2}\,\dxdx{5}{5} \ , \end{dmath*}
\begin{dmath*}\text{div } E =\left(-6\,{x_{1}}^3\,{x_{3}}^4\right)\,\d x_{1}-\frac{-1+{x_{1}}^4\,x_{2}\,{x_{3}}^5\,x_{4}}{{x_{2}}^2\,x_{3}\,x_{4}}\,\d x_{2}-\frac{-1+2\,{x_{1}}^4\,x_{2}\,{x_{3}}^5\,x_{4}}{x_{2}\,{x_{3}}^2\,x_{4}}\,\d x_{3}+\frac{1}{x_{2}\,x_{3}\,{x_{4}}^2}\,\d x_{4} \, , \end{dmath*}
\end{dgroup*}
which is $\RR$-linearly independent with the divergences obtained in Theorem \ref{MainTheorem}.
\qed
\end{proof}

\begin{rem}
Although we have not required the 2-form $F$ to be closed, it is always the case that $\d F=0$, due to the Maxwell equations. The reader could then ask whether adding this condition changes in a meaningful way the previous results. 

The answer is that both results are still valid after supposing that $\d F=0$. Let us provide a brief overview on how the proof goes: observe that a variant of Theorem \ref{MainThmMetric} that computes natural tensors associated to metrics and closed 2-forms can be proven, utilizing the techniques developed in \cite{GMN_FEDOSOV} and adding the symmetry
$$F_{ij,k}+F_{jk,i}+F_{ki,j}=0$$
to the definition of the spaces $V_r$. 

Then, the proof of Theorem \ref{MainTheorem} stays the same, taking into account that $T_{17}$ and $T_{21}$  become constant multiples of $T_{16}$ and $T_{20}$ respectively, due to the new symmetry. The example $(X,g,\bar{F})$ suffices, where $X$ and $g$ are defined in the same way and $\bar{F}:=\left({x_{1}}^2\,{x_{3}}^2\right)\,\d x_{1}\,\wedge \, \d x_{2}-\frac{2{x_{1}}^3\,{x_{3}}}{3}\,\d x_{2}\,\wedge \, \d x_{3}+\d x_{3}\,\wedge \, \d x_{4}$.

In the example $(X,g,\bar{F})$, $\div E$ is still linearly independent of the rest of divergences, and so the proof of the corollary would remain essentially unchanged.
\end{rem}

\begin{rem}
The computations in the proofs above of the expressions of the tensors and divergences in the chosen manifold, as well as their linear independence, have been performed by the mathematical software MATLAB R2020a.
\end{rem}

\noindent \textbf{Acknowledgements.} The author thanks Professor J. Navarro-Garmendia for his helpful insights during the development of this work.

\end{document}